\documentclass[a4paper,12pt]{article}
\usepackage{amssymb}
\usepackage{latexsym}
\usepackage{amsmath}
\usepackage{amsthm}
\usepackage{bm}
\usepackage{graphicx}
\usepackage{hyperref}
\usepackage[margin=1.22in]{geometry}
\usepackage{authblk}
\usepackage[round]{natbib}
\newcommand{\spayoff}{\Pi^{S}}
\newcommand{\hpayoff}{\Pi^{H}}

\newcommand{\neighb}[2]{N_{#2}(#1)}

\newcommand{\ldeg}{l}

\usepackage{color}
\newtheorem{theorem}{Theorem}

\newtheorem{lemma}{Lemma}
\newtheorem{claim}{Claim}

\newtheorem{remark}{Remark}

\theoremstyle{remark}

\definecolor{Red}{rgb}{1,0,0}
\begin{document}

\title{A Game of Hide and Seek  in Networks\footnote{Bloch's research was supported by Agence Nationale de la Recherche though grant ANR-18-CE26-0020-01. Dziubi\'{n}ski's 
 work was supported by Polish National Science Centre through grants 2014/13/B/ST6/01807 and 2018/29/B/ST6/00174. }}

\author{Francis Bloch}
\affil{Universit\'{e} Paris 1 and Paris School of Economics \authorcr 48 Boulevard Jourdan \authorcr 75014 Paris, France \authorcr \texttt{francis.bloch@univ-paris1.fr}}
\author{Bhaskar Dutta}
\affil{University of Warwick and Ashoka University \authorcr CV4 7AL Coventry, UK \authorcr \texttt{b.dutta@warwick.ac.uk}}
\author{Marcin Dziubi\'{n}ski}
\affil{Institute of Informatics, University of Warsaw \authorcr Banacha 2, 02-097 \authorcr Warsaw, Poland  \authorcr \texttt{m.dziubinski@mimuw.edu.pl}}

\date{October 2019.}

\maketitle

\begin{abstract}
We propose and study a strategic model of hiding in a network, where the network designer chooses the links and his position in the network facing the seeker who inspects and disrupts the network. We characterize optimal networks for the hider, as well as equilibrium hiding and seeking strategies on these networks. We show that optimal networks are either equivalent to cycles or variants of a core-periphery networks where every node in the periphery is connected to a single node in the core.
\end{abstract}

\newpage

\section{Introduction}
\label{sec:intro}

This paper analyses  network design problem of (say) the leader of a covert organisation who has to construct a network connecting members of her organisation as well as choose her own (hiding)  position in the network. She faces an adversary who can \lq\lq attack" one node in the network so as to catch the leader and disrupt the network. 
Related issues   have  a very long standing. According to Greek mythology, Daedalus invented the Labyrinth in order to hide the monstrous Minotaur.\footnote{See Book 8 in Ovid's Metamorphosis.} Tunnels and underground chambers in Medieval castles and fortresses were built to hide treasures or prisoners. Underground fortifications were constructed in the XXth century to hide weapons and combatants. In modern days, criminals and terrorists build covert networks in order to hide leaders, money or secret instructions.

When an object or a person is being hidden, it must also be accessible for those who need it. The Minotaur cannot be sealed off in the Labyrinth, because every nine years, he receives a tribute of seven young boys and seven young girls from Athens. The medieval treasures and prisoners, the weapons and combatants of military forts also need to be recovered and freely moved. Leaders of criminal and terrorist organizations, secret plans and money must also be able to freely and efficiently move in the network. Hence, the design of networks to hide always involves a trade-off between security (the inviolability of the hiding place) and connectivity (the accessibility of hidden objects and persons). We show below that our network design problem also exhibits this trade-off, and go on to characterise  an optimal network design as a function of this trade-off between security and efficiency. 

We construct a zero-sum game with two players, a Hider and a Seeker. In the first stage of the game, the Hider designs a network which is observed both by the Hider and Seeker. In the second stage of the game, the Hider and Seeker simultaneously choose a node in the network (where the Hider hides and the Seeker seeks or attacks). The Seeker is able to observe any  node that is a neighbour of the node attacked by her. If the Hider hides in any of the nodes observed by the Seeker, the Seeker \lq\lq captures"  her and obtains a penalty from the Hider. If the Seeker does not find the Hider, she is still able to disrupt part of the network, by removing the node that she attacks. The Hider then receives a payoff which is an increasing function of the size of the component in which she hides. The payoff in the zero-sum two-person game thus consists of two elements: (i) a benefit (to the Seeker) of capturing the hidden object or person and (ii) a benefit (to the Hider) of using a network connecting a given number of nodes. 

We characterise optimal network architectures chosen by the Hider. The optimal network  can only take one of two forms:  either it contains a cycle (where all nodes are connected in a circle) or is a special core-periphery network where half of the nodes form an interconnected core, and the other half are leaves, each connected to a single node in the core.\footnote{If the number of nodes in the core-periphery network is odd, the architecture is slightly different, with three core nodes node connected to a periphery node.} In addition, a subset of the nodes will remain isolated. The number of  isolated nodes, and the choice between the circle and the core-periphery network for connected nodes depends on the parameters of the game, and in particular the shape of the function mapping the size of the network into the benefit of the Hider. Moreover, in the cases where non-singleton nodes form a core-periphery network, the characterisation of optimal network we obtain is complete.

To understand this characterisation of an optimal network, notice that any network which cannot be ``disrupted'' (in the sense that the network is not broken into different components if the Seeker fails to find the hidden object) must be two-connected, and hence contain a cycle. Now, adding links to the cycle can only increase the sizes of the neighborhoods and hence the probability that the hidden object is discovered.\footnote{Notice however that adding links may not change the probability of capture, if the hider only hides in a subset of the nodes in the cycle.} Therefore, if the objective of the Hider is primarily to avoid disruption of the network, forming a cycle will be an optimal choice for the hider. Notice however that in a cycle, every agent has two neighbors, so the probability of discovery of the hidden object must be at least equal to $\frac{3}{n}$. In order to reduce this probability of discovery, while keeping the network connected, one has to allow for the possibility that some nodes only have degree one. In the core-periphery network where half of the nodes are leaves connected to one node in the core, the probability of discovery is reduced to the minimal value for a connected graph. In equilibrium, the Hider chooses to hide in any of the peripheral nodes, whereas the Seeker seeks in any of the core nodes. This uniform hide and seek strategy results in a probability of discovery equal to $\frac{2}{n}$, lower than in the cycle, but induces a larger disruption, as the size of the remaining component after the Seeker fails to find the object is equal to $n-2$ rather than $n-1$. In the main characterization Theorem, we show that no other network performs better than the cycle or the core-periphery network. The cycle is preferred when the Hider puts more weight on avoiding disruption and the core-periphery network is preferred when the Hider puts more weight on avoiding discovery of the hidden object.

While no real network has the exact architecture of a cycle or core-periphery network, our results echo some observations on the trade-off between security and efficiency in physical networks of military fortifications and human networks of criminals and terrorists. 

Following the trench warfare of World War 1, the French army built the ``Maginot line'', a system of underground fortifications to protect the border between Germany and France between 1929 and 1935.\footnote{Ironically, the Maginot line proved useless during the German invasion of France in May 1940, as the German army simply by-passed the line of fortifications and entered France from Belgium and Luxembourg.} The design of the underground tunnels struck a balance between separating blocks (where combatants could hide) and allowing for easy communication of men and materials. Figure \ref{fig:maginot} provides an example of the underground tunnels in three of the largest fortifications of the Maginot line: the Hackenberg, Mont des Welches and Fermont ``gros ouvrages''. It shows that blocks are not directly connected to each other (echoing the fact that peripheral nodes are only connected to one node in the core and not to each other nor to a central node), while central areas (where men sleep and weapons and ammunition are stored) form a well-connected core in the middle of the ``gros ouvrage''.

\begin{center}
\begin{figure}[h]
  \includegraphics[width=6in,height=4in]{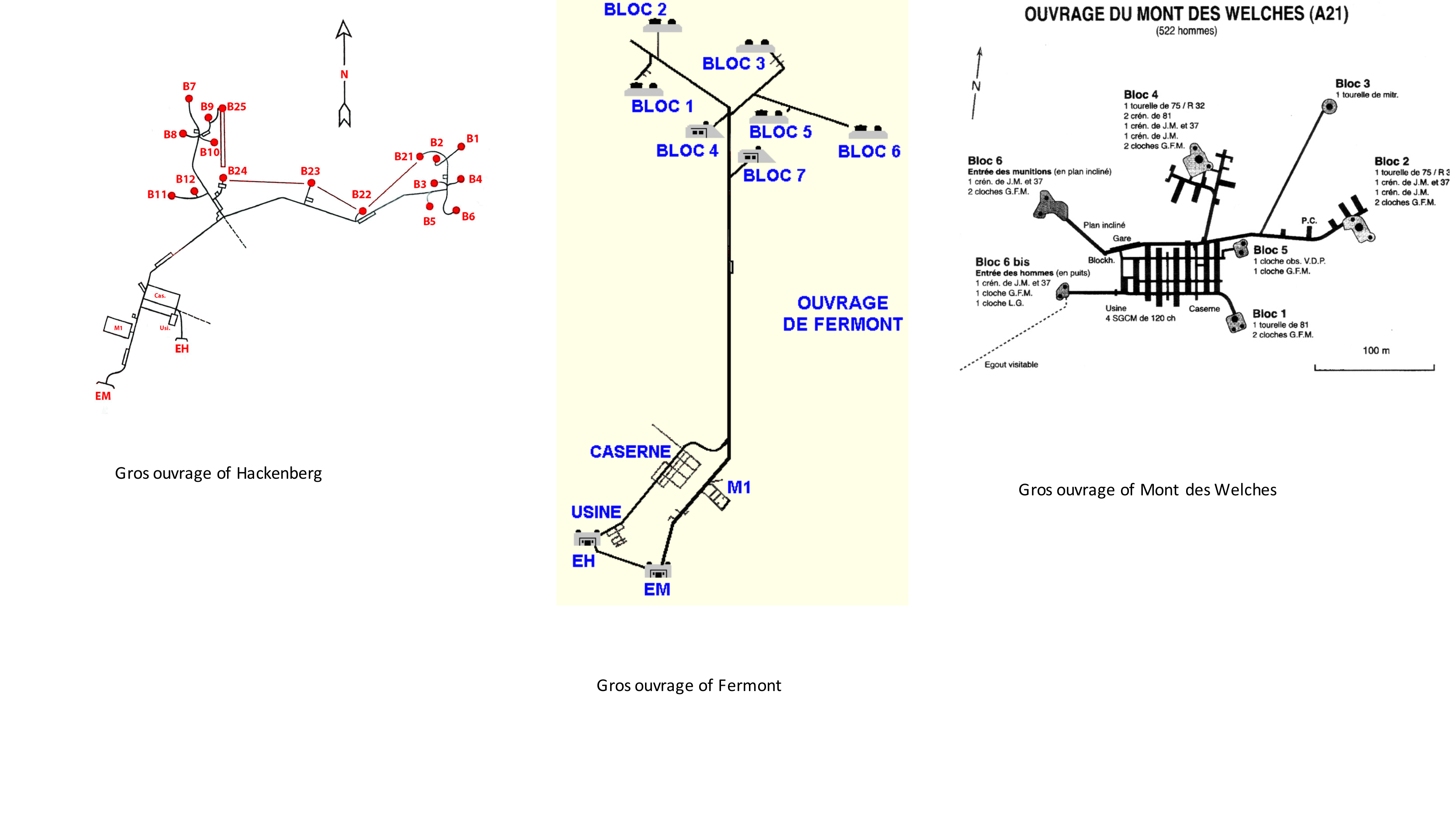}
  \caption{Three ``Gros Ouvrages'' of the Maginot Line}
  \label{fig:maginot}
\end{figure}
\end{center}
 
\cite{MorselliGiguerePetit2007} illustrate the trade-off between security and efficiency using data on terrorist networks (\cite{Krebs2002}'s map of the 9/11WTC terrorist cells) and criminal networks (a drug-trafficking network in Canada). They argue that terrorist networks are more likely to have longer average distances and fewer connections with no node assuming a central position, whereas criminal networks are more clustered and exhibit a core of nodes with high centrality. In addition they note that support nodes (which are not direct perpetrators of criminal or terrorist activities) help connect distant nodes in terrorist networks but not in criminal networks, where each support agent is attached to a single agent in the core. These two network architectures (long lines and core-periphery with clusters) can be related to the cycle and the core-periphery network we identify in our analysis. Figure \ref{fig:terrorcrime} illustrates these network architectures, by reproducing the map of the 9/11 WTC terrorist network (\cite{Krebs2002}) as well as the maps of two drug-trafficking mafia groups collected by \cite{Calderoni2012}.

\begin{center}
\begin{figure}[h]
  \includegraphics[width=6in,height=4in]{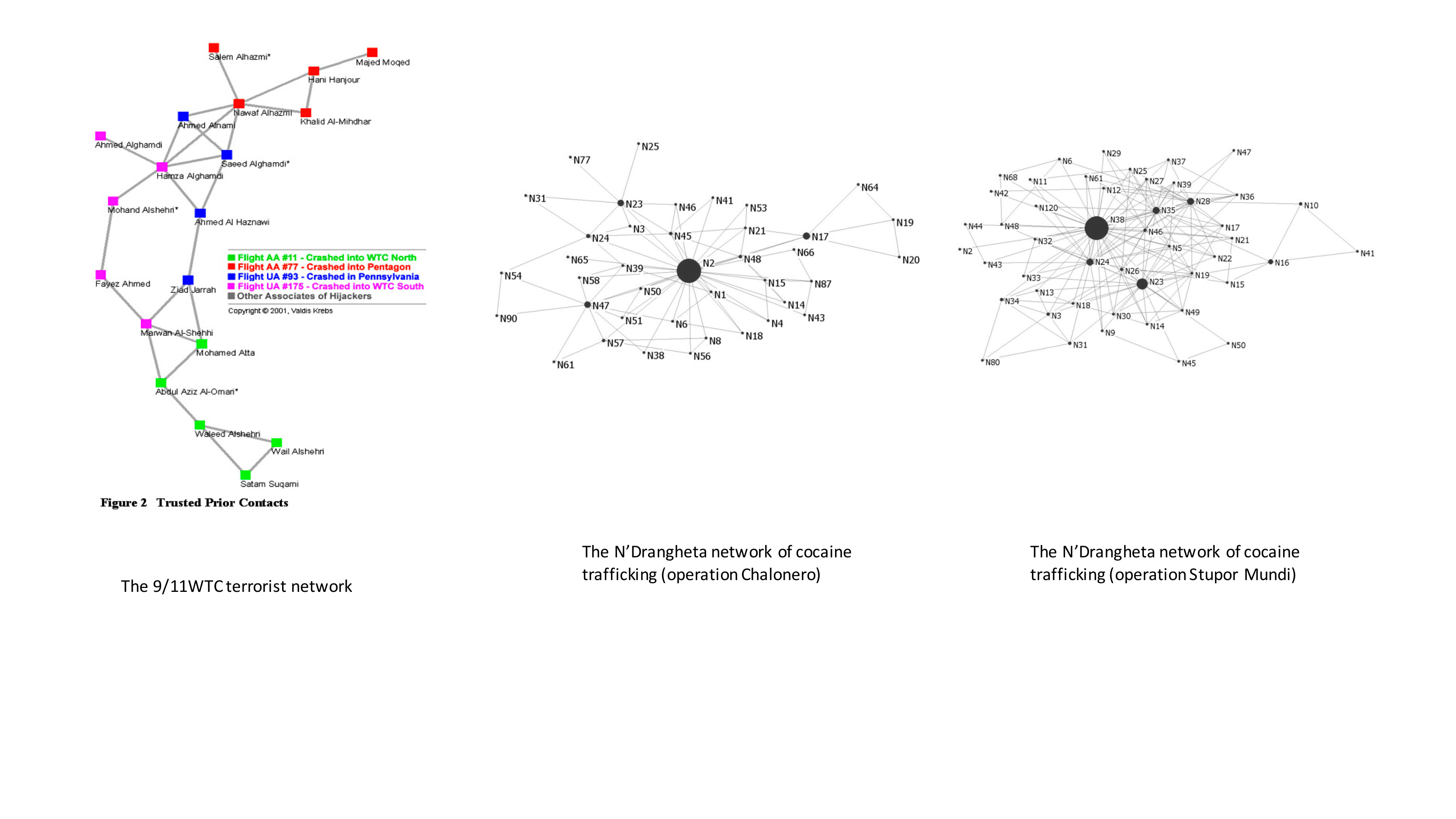}
  \caption{Three examples of terrorist and criminal networks}
  \label{fig:terrorcrime}
\end{figure}
\end{center}

\section{Related literature}

The  related literature spans a variety of disciplines, with the earlier literature focusing more on the aspect of hiding and seeking. Perhaps, the first paper was by von Neumann(1953) who discusses a zero-sum game where $H$  chooses a cell of an exogenously given matrix, 
while $S$ simultaneously  chooses a column or row in the matrix.  $S$ ``captures'' $H$ if the cell chosen by $H$ lies in the row or column chosen by $ S$. A related paper is 
Fischer (1993), who too analyses a similar zero-sum game, where $ H$ and $S$ simultaneously choose vertices of   an exogenously given graph. 
$ H$ is caught if $S$  chooses the same node as him or a node
connected to the node chosen by him.  Interestingly, the value of this ``hide and seek game'' on a fixed arbitrary network can been computed following \cite{Fisher1991}, using {\it  fractional graph theory}. \footnote{ (See also Theorem 1.4.1 in \cite{ScheinermanUllman1997})}.

Computer scientists have also  contributed to this literature.  with  \cite{WaniekMichalakRahwanWooldridge2017} and \cite{WaniekMichalakWooldridgeRahwan2018}  studying  a related, but different problem, of hiding in a network. They consider the leader of a terrorist or criminal organization, and ask the following question: How can a set of edges be added to the network in order to reduce the leader's measure of centrality  in order to avoid detection?   \cite{WaniekMichalakRahwanWooldridge2017} show that, both for degree and closeness centrality,  the problem is NP-complete. However, they also propose a procedure to build a new network from scratch around the leader (the ``captain network'') which achieves low levels of degree and closeness centrality but high values of diffusion centrality, where diffusion centrality is measured  using the independent cascade and linear threshold diffusion models. \cite{WaniekMichalakWooldridgeRahwan2018} extend the analysis to betweenness centrality and to the detection of communities (rather than individuals) in the network. Notice, however, that these models are not fully strategic since $S$ does not best respond to  $H$'s strategy.

Our paper is also related to a recent strand of  the economics literature analyzing network design and attack and defense on networks. \cite{BaccaraBarIsaac2008} study network design by an adversary (a criminal organization) taking the detection strategy of the defender as fixed. They highlight differences between two forms of detection, one which depends on the cooperation between criminals and the other which does not. In both situations, they characterize the optimal network architecture of the criminal network, which either consists of isolated two-player cells (with independent detection) or an asymmetric structure with one agent serving as an information hub (with cooperation-based detection). \cite{GoyalVigier2014} propose an alternative model of network design where the defender designs the network and chooses the distribution of defense across nodes before the attacker chooses to attack. Nodes are captured according to a Tullock contest function given the resources spent by the attacker and the defender. If a node is captured by the attacker, contagion occurs and the attacker starts attacking neighboring nodes while the defender loses his defense resources. The main message of \cite{GoyalVigier2014} is that the defendant optimally forms a star and concentrates all the defenses at the hub. \cite{DziubinskiGoyal2013} analyze a related model, where the defender designs the network and chooses defense resources before the attacker attacks. As opposed to \cite{GoyalVigier2014}, contagion does not occur and the network structure only matters through the payoffs of the two-person zero-sum game between the defender and the attacker.  The objective function of the defender is assumed to be increasing and convex in the size of components of the network, reflecting the fact that the defender wants to avoid disruption in the network.  The analysis shows that the designer will either form a star and protect the hub, or not protect any node and choose  to form a $(k+1)$-connected network when the attacker has $k$ units, so that the attacker will not be able to disrupt the network. In the same model,  \cite{DziubinskiGoyal2017} study equilibrium strategies of the defender and attacker for any arbitrary network structure while \cite{CerdeiroDziubinskiGoyal2017}  consider decentralized defense decisions by the different nodes in the network. 

The main difference between our paper and the recent literature on design, attack and defense stems from the fact that one particular node - the location chosen by the Hider- has a special significance in our model. Moreover, the capture of this node has a significant impact on payoffs. In other words, our model also incorporates the hiding-and-seeking  aspect  that is missing from the recent literature. 
 Another difference comes from the timing of the game. We suppose that the hider and seeker simultaneously choose the nodes in which to hide and that they inspect, resulting in equilibria in  mixed strategies as in Colonel Blotto games, whereas \cite{GoyalVigier2014} and \cite{DziubinskiGoyal2013} assume that the defender and attacker move sequentially, allowing for pure strategy equilibria. Note also that nodes cannot be defended in our analysis.

\section{The Model}
\label{sec:model}
There are two players, a \emph{Hider} ($H$)  and a \emph{Seeker} ($S$).  The hider $H$  constructs a network among $n$ nodes and chooses a location in the network. For example, the Hider may be the leader of a covert  terrorist or criminal organisation, which  has  $n-1$ other members.   The seeker is then interpreted as a law enforcing agency whose objective is to capture the leader of the organization or to disrupt the communication channels within the organisation. The interaction between $H$ and $S$ is modeled as a two-stage process, which is described below. 

In the first stage, $H$ chooses a network of interactions amongst the members of the organisation. 
Formally, $H$  chooses a graph $G = \langle V, E \rangle$ where $V$ is a set of $n$ vertices,  and $E$ is a  set of undirected edges $E \subseteq \binom{V}{2}$. A typical edge $ e\in E$ will be denoted $ij$, where $i, j \in V$. 

 Both players observe the chosen network at the beginning of  the second stage.   After observing the network $G$,  players $H$ and $S$ {\it simultaneously} choose {\it one} node each.   
The node chosen by the hider is his (hiding) position in the network. The node chosen by the seeker is the node she inspects (or attacks). Let $k$ be the node chosen by  $S$, and $N_G(k) = \{ j \in V| kj \in E\}$.  That is, $N_G(k)$ is the set of all neighbours of $k$ in $G$. All nodes in $\{k\} \cup N_G(k)$ can be observed by the seeker. If the chosen  position of $H$ is in $\{k\} \cup N_G(k)$, then $H$ is captured by $S$.  In addition, node $k$ is removed from the network, irrespective of whether $H$ is captured or not. 

The seeker uses his choice to capture the hider and to damage the network. Payoffs depend on whether or not the hider has been captured. If caught, the hider gets payoff $-\beta$, where $\beta \geq 0$. 

If the hider is not captured, the covert network remains operational, but is damaged by the attack of the seeker. Then the hider's payoff depends on the size of the component he is hiding in in the residual network. Formally,  his payoff is given by a  function $f : \mathbb{R}_{\geq 0} \rightarrow \mathbb{R}_{\geq 0}$ of the size of his component in the residual network. We assume $f$ to be strictly increasing with $f(0) = 0$.   An example of function $f$ in line with these assumptions is the identity function, $f(x) = x$ for all $x \in \mathbb{R}_{\geq 0}$. The game is assumed to be  a zero-sum game, so that  the  payoff to the seeker is equal to minus the payoff of the hider.

Formally, given a set of nodes $U \subseteq V$, let $\mathcal{G}(U)$ be the set of all undirected graphs over $U$ and let $\mathcal{G} = \bigcup_{U\subseteq V} \mathcal{G}(U)$ be the set of all undirected graphs that can be formed over $V$ or any of its subsets.
A strategy for the hider is a pair $(G,h) \in \mathcal{G}(V) \times V$, where $G$ is the graph and $h$ is the hiding place chosen by $H$ in $G$. As the seeker chooses his inspected node after observing the network, a strategy for the seeker is a function $s : \mathcal{G}(V) \rightarrow V$. 

Before defining the payoffs we introduce some auxiliary definitions on networks.
Given a set of nodes $U\subseteq V$ and a graph $G = \langle U, E \rangle$ over $U$, a maximal set of nodes $C \subseteq U$ such that any two nodes $i,j \in C$ are connected in $G$ is a \emph{component} of $G$.\footnote{
Two nodes $i,j \in U$ are connected in $G = \langle U, E\rangle$ if there exists a sequence of nodes $i_1,\ldots,i_l$ such that $i_0 = i$, $i_n = j$, and for all $k \in \{1,\ldots,l\}$, $i_{k-1}i_k \in E$.} The set of all components of $G$ is denoted by $\mathcal{C}(G)$. In addition, given $i \in U$, let $C_i(G)$ be the   component in $G$ containing $i$. Given a set of nodes $U\subseteq V$, a graph $G = \langle U, E \rangle$ over $U$, and a set of nodes $U'\subseteq U$, let $G[U'] = \langle U', E[U'] \rangle$ with $E[U'] = \{ij \in E : \{i,j\}\subseteq U'\}$ be the \emph{subgraph of $G$ induced by $U'$}. Given a node $k \in V$ let $G - k = G[U\setminus \{k\}]$ be the \emph{residual network} obtained from $G$ by removing $k$ and all its links from $G$.

Given the strategy profile $((G,h),s)$, the payoff to the hider is
\begin{equation}
\hpayoff(G,h,s) = \left\{\begin{array}{ll}
                       -\beta & \textrm{if $h\in \{s(G)\} \cup \neighb{s(G)}{G}$} \\
                       f(|C_i(G - s(G))|) & \textrm{otherwise.}
                       \end{array}
                \right.
\end{equation}
The payoff to the seeker is $\spayoff((G,h),s) = -\hpayoff((G,h),s)$.

The cycle network and the {\it core-periphery } networks will be important in our analysis. The cycle network is the unique network where every node has exactly two neighbors.

A \emph{core-periphery} network over a set $V = P \cup C$ of $n$ nodes is  defined
as follows. There are $q \geq \lceil n/2 \rceil$ \emph{core} nodes in set $C = \{\mathrm{c}_1,\ldots,\mathrm{c}_q\}$ and $m \leq \lfloor n/2 \rfloor$ \emph{periphery} nodes in set $P = \{\mathrm{p}_1,\ldots,\mathrm{p}_{m}\}$. Nodes of the core form a connected graph, while each periphery node, $\mathrm{p}_i$ with $1 \leq i \leq m$, is connected to core node $\mathrm{c}_{i}$.
Nodes of the core which are not connected to a periphery node are called \emph{orphaned} . Figure~\ref{fig:coreperipherymain} illustrates a core-periphery network with orphaned nodes.

\begin{figure}[!htb]
\centering
\includegraphics[scale=.75]{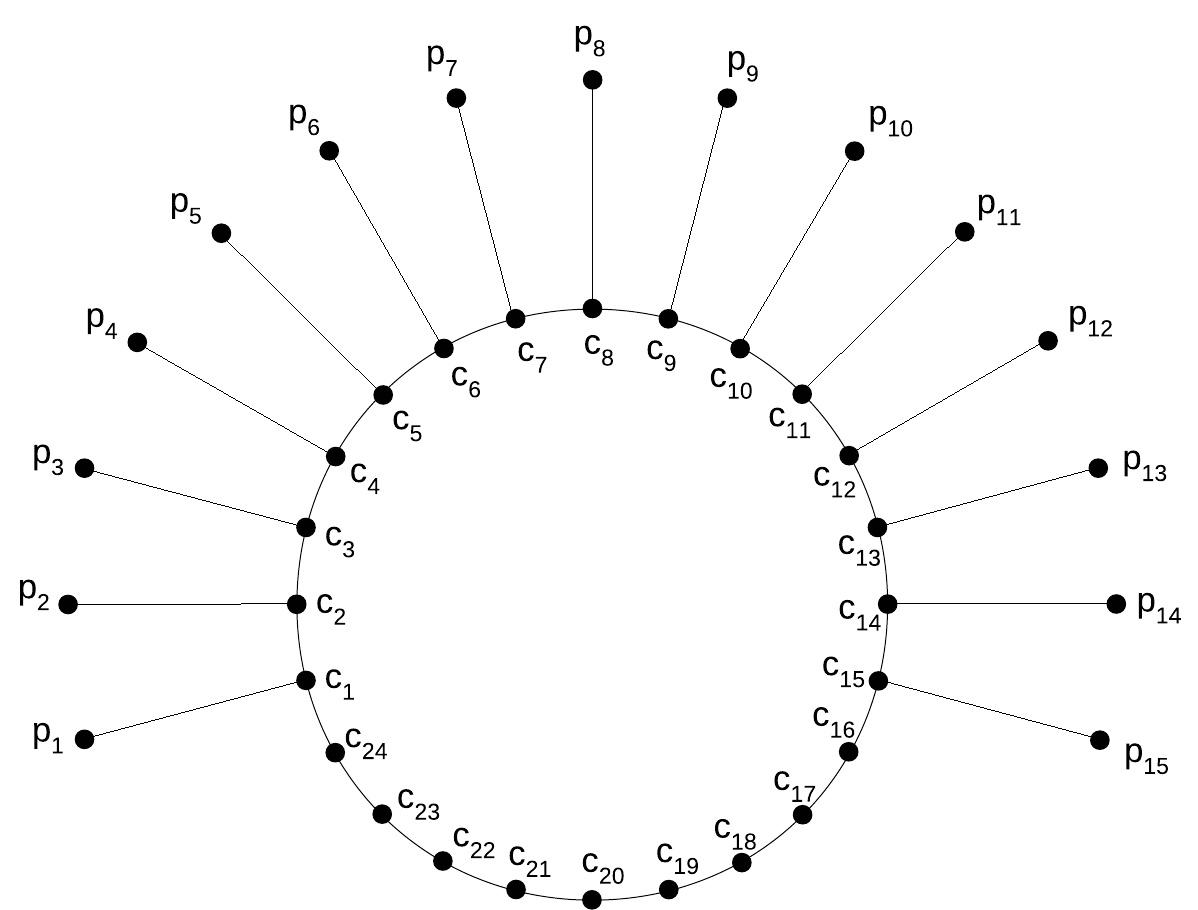}
\caption{A core-periphery network over $39$ nodes, with $15$ periphery nodes and $9$ orphaned core nodes.}
\label{fig:coreperipherymain}
\end{figure}

A particular class of core-periphery networks, which we call \emph{maximal}, plays a crucial role in our characterization. If $n$ is even, a core-periphery network is maximal if and only if it has $n/2$ periphery nodes and nodes of the core form a $2$-connected graph.\footnote{
A graph is $2$-connected if and only if it does not get disconnected after removing a single node.}
If $n$ is odd, a core-periphery network is maximal if and only if it has $(n-3)/2$ periphery nodes (and hence $3$ orphaned nodes), nodes of the core form a $2$-connected graph, and one orphaned node has exactly the two other orphaned nodes as its neighbours. Examples of maximal core-periphery networks are presented in Figure~\ref{fig:maxcp}.

\begin{figure}[!htb]
\centering
\includegraphics[scale=.75]{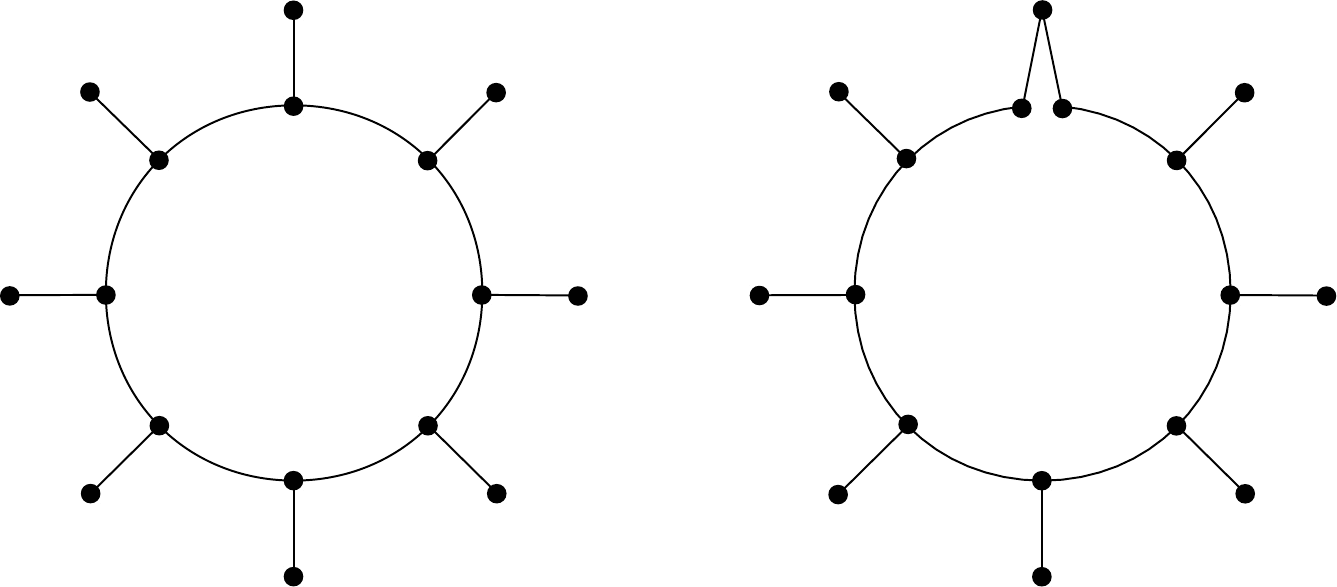}
\caption{Maximal core-periphery networks over $16$ nodes (left) and $17$ nodes (right).}
\label{fig:maxcp}
\end{figure}

\section{The Characterization Result}

Our objective in this section is to provide optimal networks for the hider as well as to characterise the hiding and the seeking strategies on these networks. 
 We show in our main result (Theorem~\ref{th:singlemain}) that these networks consist of a number of singleton nodes and a connected component which  contains either a {\it cycle} or has a particular {\it core periphery} topology.  

Whether a cycle or a core-periphery topology is better for the hider depends on the value of the following expression.
\begin{equation}
T(n,s) = (n - s - 3)f(n-s-1) - (n - s - 2)f(n-s-2).
\end{equation}
We will show that the cycle topology is better if $T(n,s) > \beta$, while  a core-periphery topology is better when $T(n,s) < \beta$. 

The following lemma asserts that in an  optimal network there cannot be a component  containing just two or three nodes. The lemma will be used in the proof of the main theorem.
\begin{lemma}\label{sing}
Suppose $G$ is an optimal network for $H$ whose set of non-singleton components is $\mathcal{X}$. Then, each component $C \in \mathcal{X}$ contains at least $4$ nodes.
\end{lemma}

\begin{proof} 
Suppose the lemma is not true and some $C \in \mathcal{X}$ has exactly three nodes, $C = \{n_1, n_2, n_3\}$.  Following standard arguments, $C$ must have a non-empty intersection with the support of $H$'s optimal hiding strategy as well as $S$'s optimal seeking strategy, given $G$. (If not, the hider or the seeker would have profitable deviations). Moreover,  conditional on hiding in $C$, $H$ is caught with probability $\rho$, where $\rho$ is the total probability with which $S$ seeks in $C$. This is true because $S$ can always search one node in $C$ that has two neighbours, and hence observe all nodes in the component.

Let $G'$ be another network which coincides with $G$ everywhere except that $C$ is broken up into singleton components $\{n_1\}$, $\{n_2\}$, $\{n_3\}$.   Moreover, suppose $H$'s hiding strategy coincides  with that in $G$  everywhere on $V\setminus C$, while $H$ distributes the earlier probability  weight on $C$ uniformly on the three nodes $n_1$, $n_2$, $n_3$.  It is straightforward to check that  $H$'s expected payoff in $G'$ is strictly higher than his expected payoff in inspecting a node in the component $C$, which must be equal to his expected payoff of using a mixed strategy in $G$, contradicting optimality of $G$.

A similar argument rules out an  optimal network containing  a component with only two nodes.
\end{proof}

\begin{remark}
\label{singr}
An implication of this lemma  is that the optimal network will either be completely disconnected with $n$ singletons or will contain at most $n-4$ singletons.  This implication will be used throughout the proof of the theorem.\end{remark}

At this stage, we describe the main result of the paper somewhat  informally. A formal statement appears towards the end of the section. 

We will construct an equilibrium that will have the following features. 
\begin{itemize}
\item The optimal network $G$  will have a certain number of singleton nodes $s$ (that will be determined) where $s \le n-4$ or $s =n$.
\item If $T(n, s) \ge \beta$ and $s\ne n$, then $G$ has a cycle component over $n-s$  nodes. 
\item If $T(n,s) < \beta$, $n - s \geq 4$, then $G$ will have a maximal core periphery component over $n-s$ nodes.
\item The hider mixes between hiding in the singleton nodes and in the connected component with probabilities that will be determined.  When hiding in the singleton nodes, he mixes uniformly across all these nodes. When hiding in the connected component, he mixes uniformly across all the nodes when it is a cycle, mixes uniformly across the periphery nodes when it is a maximal core-periphery network over even number of nodes, and mixes between hiding in periphery nodes, mixing uniformly across them, and the middle orphaned node, otherwise.
\item The seeker mixes between seeking in the singleton nodes and in the connected component. When seeking in the singleton nodes, he mixes uniformly across all these nodes. When seeking in the connected component, he mixes uniformly across all the nodes when it is a cycle, mixes uniformly across the core nodes when it is a maximal core-periphery network over even number of  nodes, and mixes between seeking in the neighbours of periphery nodes, mixing uniformly across them, and the middle orphaned node, otherwise.
\end{itemize}

To get some intuition behind the result, notice that the hider faces a tradeoff between the cost of being caught and the value he gets in the residual network.  Adding links in the network increases connectivity and hence  secures a larger value after the the seeker's action {\it provided} he is not caught.  However, a larger number of links also leads to higher exposure and a greater probability of being caught, as it increases the size of the neighborhoods of the nodes in which the hider can hide.
Fixing the number of singleton nodes, $s$, the choice between a cycle and a core-periphery network is influenced by the change in $f$, as measured by the quantity $T(n,s)$. The probability of being caught in a cycle of size $n-s$ is $3/(n-s)$, as each node has exactly two neighbours, while  only one node is lost from the cycle component if not caught. The probability of being caught in a maximal core-periphery network (if $n-s$ is even), on the other hand, is $2/(n-s)$ since the hider hides mixing uniformly across the periphery nodes; in the event of not being caught, two nodes are lost from the core periphery  component  since the seeker seeks mixing uniformly across the core nodes.
If the change in $f$ between $n-s-2$ and $n-s-1$ is sufficiently large, so that $T(n,s) > \beta$ then the marginal loss from an additional node being removed from a component is high, as compared to the penalty for being caught, and, therefore, a cycle is preferred over the core-periphery network. If the change in $f$ is not sufficiently large, on the other hand, the marginal loss from an additional node being removed from a component is not sufficiently high and the hider prefers to opt for the safer, core-periphery, network. 

The proof of the theorem is long and we provide a brief description of the general technique before giving the details.

We start by constructing a feasible strategy of the seeker that, for each network over the set of nodes $V$, provides a (mixed) seeking strategy on that network. This strategy determines the payoffs the seeker can secure for each possible network over $V$. Since the game is zero-sum, minus these payoffs provide  an upper bound on the payoff the hider can get for each network. Next, for each $s\in \{0,\ldots,n-4, n\}$, we construct a network that is optimal for the hider across all possible networks with exactly $s$ singleton nodes. In the case of $T(n,s)\ge \beta$, as well as in the case of $s$ being even, these networks yield payoffs to the hider that meet the upper bound determined in the first part of the proof. In the case of $T(n,s) < \beta$ and odd $s$, the upper bound from the first part of the proof is not exact. Therefore in this step we establish both the optimal networks for the hider and the exact upper bound on the hider's payoff.

We will use a series of lemmas to prove the theorem.  We first introduce a partition of the nodes into  different sets that will play a crucial role in the construction of a strategy for the seeker.

Given a (possibly disconnected) network $G$ over the set of nodes $V$, node $i \in V$ is a \emph{singleton node} if $|\neighb{i}{G}| = 0$.
The set of singleton nodes of $G$ is denoted by $S(G)$.
Node $i\in V$ is a \emph{leaf} if $|\neighb{i}{G}| = 1$. The set of leaves of $G$ is denoted by $L(G)$. Given node $i \in V$, let $\ldeg_i(G) = |\neighb{i}{G}\cap L(G)|$ denote the number of leaf-neighbours of $i$. 

Let 
\begin{equation*}
M(G) = \{i \in V : \ldeg_i(G) = 1 \}
\end{equation*}
be the set of nodes which  are connected to exactly one leaf in $G$ and let 
\begin{equation*}
S\!L(G) = \{i \in L(G) : \neighb{i}{G} \cap M(G) \neq \varnothing \}
 \end{equation*}
be the set of leaves connected to an element of $M(G)$. Such leaves are called \emph{singleton leaves}.
Let $R(G) = V \setminus (S(G) \cup S\!L(G) \cup M(G))$ be the set of nodes in $G$ which are neither a singleton, nor a singleton leaf, nor a neighbour of a singleton leaf.

We now construct a strategy for the seeker which guarantees a given payoff for any network $G$.
Take any network $G$ over $V$ and let $s = |S(G)|$ and $m = |M(G)|$.  Moreover, let $G\!R = G[R(G)]$ be the subnetwork of $G$ generated by the set of nodes $R(G)$. In particular, when $R(G) = \varnothing$, $G\!R$ is the empty network with empty sets of nodes and links. Let $D(G\!R)$ be the set of nodes in $R(G)$ that belong to two-element subsets of $R(G)$. 

Consider a mixed strategy  of player $S$, $\bm{\sigma} = (\sigma_1,\ldots,\sigma_n)$, with the following probabilities:
\begin{equation}
\label{eq:sigma}
\bm{\sigma} = \lambda_{\mathrm{S}} \bm{\sigma}^{\mathrm{S}} + \left(1 - \lambda_{\mathrm{S}}\right)\left( \lambda_{\mathrm{R}} \bm{\sigma}^{\mathrm{R}} + \left(1-\lambda_{\mathrm{R}}\right) \bm{\sigma}^{\mathrm{M}}\right)
\end{equation}
where  $\lambda_{\mathrm{R}},\lambda_{\mathrm{S}} \in [0,1]$, and
\begin{align*}
\sigma_i^{\mathrm{S}} & = \begin{cases}
                           \frac{1}{s}, & \textrm{if $i \in S(G)$,}\\
                           0, & \textrm{otherwise,}
                           \end{cases}\\
\sigma_i^{\mathrm{M}} & = \begin{cases}
                           \frac{1}{m}, & \textrm{if $i \in M(G)$,}\\
                           0, & \textrm{otherwise,}
                           \end{cases}\\
\sigma_i^{\mathrm{R}} & = \begin{cases}
                           \frac{l_i(G\!R)+1}{n-s-2m}, & \textrm{if $i \in R(G)\setminus (L(G\!R)$,}\\
                           \frac{1}{{n-s-2m}}, & \textrm{if $ i \in D(G\!R)$,}\\
                           0, & \textrm{otherwise,}
                           \end{cases}
\end{align*}

We first show that these probabilities are well-defined.
\begin{lemma}\label{feas}
  $\bm{\sigma}$ is a feasible strategy for the seeker $S$.
\end{lemma}

\begin{proof}
Clearly, $\bm{\sigma}^{\mathrm{S}}$ is a valid probability distribution as long as $S(G)\neq \varnothing$, that is $s > 0$.
Similarly, $\bm{\sigma}^{\mathrm{M}}$ is a valid probability distribution as long as $M(G) \neq \varnothing$, that is $m \geq 1$. It is also easy to see that $\bm{\sigma}^{\mathrm{R}}$ is a valid probability distribution as long as $R(G)\neq \varnothing$. To see this,  notice that $R(G)$ contains exactly $n - s - 2m$ nodes and $\bm{\sigma}^{\mathrm{R}}$ can be obtained from a uniform distribution on $R(G)$ by moving the probability mass assigned to leaves in $G\!R \setminus D(G\!R)$ to    their neighbours. Lastly, notice that if $S(G) \neq \varnothing$, then either all the non-singleton nodes in $G$ have degree $1$, in which case $M(G)\neq \varnothing$, or there exists a node in $G$ of degree $2$ or more, in which case either $M(G)\neq \varnothing$ or $R(G)\neq \varnothing$. Hence if $S(G) \neq \varnothing$, then either $\bm{\sigma}^{\mathrm{M}}$ or $\bm{\sigma}^{\mathrm{R}}$ is a valid probability distribution.
 By these observations, $\bm{\sigma}$ is a valid probability distribution as long as $\lambda_{\mathrm{S}} = 1$,  if $s =n$, $\lambda_{\mathrm{S}} = 0$, if $s = 0$, $\lambda_{\mathrm{R}} = 0$, if $R(G)= \varnothing$, and $\lambda_{\mathrm{R}} = 1$, if $m = 0$. 
So, the lemma is true.
\end{proof}

The idea behind the construction of strategy $\bm{\sigma}$ is as follows. With probability $\lambda_{\mathrm{S}}$, player $S$ seeks in the set of singleton nodes, $S(G)$, and with probability $(1-\lambda_{\mathrm{S}})$ he seeks outside this set. Conditional on seeking outside $S(G)$, with probability $\lambda_{\mathrm{R}}$ player $S$ seeks in the set of nodes $R(G)$ and with probability $(1 - \lambda_{\mathrm{R}})$ he seeks in the set $S\!L(G) \cup M(G)$.
When seeking in $S(G)$, $S$ mixes uniformly across all the singleton nodes. When seeking in $SL(G)\cup M(G)$, $S$ mixes uniformly across all the nodes neighbouring a singleton leaf, that is all the nodes in $M(G)$. Lastly, when seeking in the set of nodes $R(G)$, $S$ mixes using strategy $\bm{\sigma}^{\mathrm{R}}$.  In the next two lemmas, we compute lower bounds on the probability of capture of the hider in different regions of the network.

\begin{lemma}\label{capt1}
The probability of capture of player $H$ 
is at least $(1-\lambda_{\mathrm{S}})\lambda_{\mathrm{R}} 3 / (n-s-2m)$, if $H$ hides in $R(G) \setminus (S(G\!R) \cup SL(G\!R) 
\cup D(G\!R))$.
\end{lemma}

\begin{proof}
 Take  any node $i\in R(G) \setminus (S(G\!R) \cup SL(G\!R) 
\cup D(G\!R))$.
Suppose, first,  that $i$ is not a leaf in $G\!R$, i.e. $i \in R(G)\setminus L(G\!R)$. Then $i$ has at least two neighbours in $R(G)$ and the probability that seeker seeks at $i$ or at one of $i$'s neighbours is at least $(1-\lambda_{\mathrm{S}})\lambda_{\mathrm{R}} 3 / (n-s-2m)$. 
Suppose, next,  that $i\in L(G\!R)\setminus (S\!L(G\!R)\cup D(G\!R))$.
Then $i$ has a neighbour $j \in R(G)$ that has at least one more leaf neighbour in $G\!R$. Since $\sigma_j  =(1-\lambda_{\mathrm{S}})\lambda_{\mathrm{R}} 3 / (n-s-2m)$, the lemma is true. 
\end{proof}

\begin{lemma}
\label{capt2}
The probability of capture of player $H$ 
is at least $(1-\lambda_{\mathrm{S}})\lambda_{\mathrm{R}} 3 / (n-s-2m)$, if $H$ hides in $ S(G\!R) \cup SL(G\!R) 
\cup D(G\!R)$.
\end{lemma}

\begin{proof}
 In this case, $i$ must have a neighbour, $j$, in $M(G)$. For otherwise $i$ would be a singleton node in $H$ or a singleton leaf in $H$ and so $i$ would belong to $S(G) \cup M(G)$ and not to $R(G)$. Now,the probability of $S$ putting a seeking resource in $j$ is
\begin{eqnarray*}
\sigma_j &= &(1-\lambda_{\mathrm{S}})(1-\lambda_{\mathrm{R}})\left(\frac{1}{m}\right)\\ &\geq &
(1-\lambda_{\mathrm{S}})\min\!\left(1,\frac{3m(f(n-s-1)+\beta)}{3m(f(n - s-1) + \beta)+(n - s - 2m)(f(n - s-2) + \beta)}\right)\!\left(\frac{1}{m}\right) \\&= &
(1-\lambda_{\mathrm{S}})\left(\frac{3(f(n-s-1)+\beta)}{3m(f(n - s-1) + \beta)+(n - s - 2m)(f(n - s-2) + \beta)}\right) \\&> &
(1-\lambda_{\mathrm{S}})\left(\frac{3(f(n-s-2)+\beta)}{3m(f(n - s-1) + \beta)+(n - s - 2m)(f(n - s-2) + \beta)}\right) \\&= &
(1-\lambda_{\mathrm{S}})\lambda_{\mathrm{R}}\left(\frac{3}{n-s-2m}\right).
\end{eqnarray*}
Thus $i$ is caught with probability at least $(1-\lambda_{\mathrm{S}})\lambda_{\mathrm{R}} 3 / (n-s-2m)$.
\end{proof}

We now use these characterisations to compute lower bounds on the expected payoff of the seeker when the hider hides in different parts of the network.

\begin{lemma}
Conditional on $H$ hiding in a node of $R(G)$ and $S$ using $\bm{\sigma}$, the expected payoff of $S$ is at least 
\begin{multline}
L^{\mathrm{R}}(n,m,s) =
\left(1-\lambda_{\mathrm{S}}\right)\bigg(\lambda_{\mathrm{R}} \bigg(\left(\frac{3}{n-s-2m}\right)\beta - \left(1 - \frac{3}{n-s-2m}\right)f(n-s-1)\bigg)\\
 -\left(1-\lambda_{\mathrm{R}}\right)f(n-s-2)\bigg) - \lambda_{\mathrm{S}} f(n-s)
 \end{multline}
\end{lemma}

\begin{proof}  Suppose that $H$ hides in $R(G)$. From lemmas \ref{capt1} and \ref{capt2}, $H$ is captured with probability at least $(1-\lambda_{\mathrm{S}})\lambda_{\mathrm{R}} 3 / (n-s-2m)$ when $S$ chooses $\bm{\sigma}$. If not captured, only one node is removed when $S$ searches in $R(G)$. With probability $(1-\lambda_{\mathrm{S}})( (1-\lambda_{\mathrm{R}})$, $S$ searches in $M(G)$ and removes two nodes. Finally, with probability $\lambda_{\mathrm{S}}$, $S$ searches in $S(G)$,  and does not catch $H$. Then, her payoff is at least $-f(n-s)$ - this happens if $G$ is connected over $n-s$ nodes. 
\end{proof}

Similarly, we  compute a lower bound on the expected payoff of the seeker when the hider hides in $M(G)$ or $SL(G)$:

\begin{lemma} 
Conditional on $H$ hiding in a node of $M(G)\cup SL(G)$, player $S$ by choosing  $\bm{\sigma}$ obtains a payoff of at least 
\begin{multline*}
L^{\mathrm{M}}(n,m,s) = \left(1-\lambda_{\mathrm{S}}\right)\bigg(\left(1-\lambda_{\mathrm{R}}\right) \bigg(\left(\frac{1}{m}\right)\beta - \\
\left(1 - \frac{1}{m}\right)f(n-s-2)\bigg) - \lambda_{\mathrm{R}}f(n-s-1)\bigg) - \lambda_{\mathrm{S}} f(n-s),
\end{multline*}
\end{lemma}
\begin{proof} 
  The probability of capture of $H$ is at least $(1-\lambda_{\mathrm{S}})(1-\lambda_{\mathrm{R}}) 1 / m$.  If $H$  is not   captured, $S$ guarantees that the component of the hider has size at most $n-s-2$ with probability $(1-\lambda_{\mathrm{S}})(1-\lambda_{\mathrm{R}})$ when the attack is in $M(G)$. Furthermore,  at least one node is removed with probability $(1-\lambda_{\mathrm{S}})\lambda_{\mathrm{R}}$ when the attack is in $R(G)$. Finally, the component containing $H$  has size at most $n-s$ when the attack is in $S(G)$, and this happens  with probability $\lambda_{\mathrm{S}}$.
\end{proof}

We now set the value of $\lambda_{\mathrm{R}}$ in order to equalize the probability of capture of the hider in different regions of the network, outside singleton nodes. To this end, we assume that  there exist non-singleton nodes in $G$, $S(G)\neq V$,.
Let
\begin{align*}
\rho & = \frac{(n-s-2m)(f(n-s-2)+\beta)}{3m(f(n - s-1) + \beta)+(n - s - 2m)(f(n - s-2) + \beta)} \\
        & = 1 - \frac{3m(f(n-s-1)+\beta)}{3m(f(n - s-1) + \beta)+(n - s - 2m)(f(n - s-2) + \beta)}
\end{align*}
and
\begin{equation}\label{lambr}
\lambda_{\mathrm{R}} = \begin{cases}
                       0, & \textrm{if $R(G)= \varnothing$,}\\
                       \rho, & \textrm{otherwise.}
                       \end{cases}
\end{equation}
Clearly $\rho \in [0,1]$ and $\lambda_{\mathrm{R}} \in [0,1]$.

It is straightforward  to verify that the chosen value of $\lambda_{\mathrm{R}}$ ensures that  $L^{\mathrm{R}}(n,m,s) = L^{\mathrm{M}}(n,m,s)$, for any $s\in \{0,\ldots,n-4\}$. 
Hence the lower bound on the payoff of player $S$ in $G$ when $H$ hides outside singleton nodes is 

\begin{equation}\label{bound1}
L(n,m,s) = L^{\mathrm{R}}(n,m,s) = L^{\mathrm{M}}(n,m,s) = \left(1-\lambda_{\mathrm{S}}\right)A(n,m,s) - \lambda_{\mathrm{S}}f(n-s)\end{equation}
where
\begin{multline*}
A(n,m,s) = \begin{cases}
         \frac{\beta}{m} - \left(\frac{m-1}{m}\right)f(n-s-2), \ \textrm{if $R(G)= \varnothing$,}\\
         \left(\frac{D(n,s)D(n-1,s)}{3D(n,s)-2D(n-1,s)}\right)\left(\frac{3(\beta - T(n,s))}{m(3D(n,s)-2D(n-1,s))+(n-s)D(n-1,s)}-1\right)
         +\beta,\\ \ \textrm{otherwise}
         \end{cases}
\end{multline*}
with
\begin{equation*}
D(n,s) = f(n-s-1)+\beta
\end{equation*}
and
\begin{equation*}
T(n,s) = (n - s - 3)D(n,s) - (n - s - 2)D(n-1,s) + \beta
\end{equation*}
In particular, the derivation above is valid for the extreme cases of $m = 0$ and $m = (n-s)/2$).
Notice that $A(n,m,s)$ is strictly increasing in $m$ if $T(n,s)>\beta$, is strictly decreasing in $m$ if $T(n,s)<\beta$, and is constant if $T(n,s) = \beta$.

To complete the definition of strategy $\bm{\sigma}$ we compute the value of the probability of seeking in singleton nodes, $\lambda_{\mathrm{S}}$.
Conditional on $H$ hiding in a node of $S(G)$, using any of the strategies $\bm{\sigma}$ defined above, player $S$ obtains payoff of at least $L^{\mathrm{S}}(n,m,s) = \lambda_{\mathrm{S}}B(s) - \left(1-\lambda_{\mathrm{S}}\right)f(1)$, where
\begin{equation*}
B(s) = \left(\frac{1}{s}\right)\beta - \left(1 - \frac{1}{s}\right)f(1),
\end{equation*}
regardless of the strategy of the hider, as the probability of capture is $\lambda_{\mathrm{S}}/s$ and, in the case of not capturing the hider, $S$ gets payoff $-f(1)$.
Let 
\begin{equation*}
\lambda_{\mathrm{S}} = \begin{cases}
                       1, & \textrm{if $s=n$,}\\
                       \frac{A(n,m,s)+f(1)}{A(n,m,s) + B(s) + f(1) + f(n-s)}, & \textrm{if $s \neq n$ and  $A(n,m,s) > -f(1)$,}\\
                       0, & \textrm{otherwise.}
                       \end{cases}
\end{equation*}
To see that $\lambda_{\mathrm{S}} \in [0,1]$, notice that $B(s) > -f(1) \geq -f(n-s)$, for any $\beta\geq 0$ and $0\leq s \leq n-4$.

It is straightforward  to verify the following for any $s\in \{0,\ldots,n-4\}$:
\begin{itemize}
\item[(i)] if  $A(n,m,s) > -f(1)$, then $L^{\mathrm{s}}(n,m,s) = L(n,m,s)$.
\item [(ii)] if $A(n,m,s) \leq -f(1)$ then $L^{\mathrm{s}}(n,m,s) \geq L(n,m,s)$.
\end{itemize}

So, finally,  if $s\leq n-4$,  the lower bound on the payoff of player $S$ in $G$ is given by

$$Q(n,m,s) = (1-\lambda_{\mathrm{S}})A(n,m,s) - \lambda_{\mathrm{S}}f(n-s),$$
Of course, if  $s=n$,  $\bm{\sigma}$ mixes uniformly across the singletons with $\lambda_{\mathrm{S}} = 1$. 

To summarize, the lower bound on the payoff of $S$ in $G$,   secured by the strategy $\bm{\sigma}$,  is given by
\begin{multline}\label{defq}
Q(n,m,s) = 
\begin{cases}
B(n), & \textrm{$s=n$} \\
\frac{A(n,m,s)B(s) - f(1)f(n-s)}{A(n,m,s) + B(s) + f(1)+f(n-s)}, & \textrm {if $s\leq n-4$ and $A(n,m,s) > -f(1)$},\\
A(n,m,s), & \textrm{otherwise},
\end{cases}
\end{multline}

Recall that $A(n,m,s)$ 
is increasing in $m$ when $T(n,s) > \beta$, decreasing in $m$ when $T(n,s) < \beta$, and constant in $m$  when $T(n,s) = \beta$. This, fact, together with Claim~\ref{claim:decr} in the appendix implies that when $s\leq n-4$, $Q(n,m,s)$ is decreasing in $m$ when $T(n,s) < \beta$, increasing in $m$ when $T(n,s) > \beta$, and is constant in $m$ when $T(n,s) = \beta$.
{\em So for all   $s \in \{0,\ldots,n-4\}$, $Q(n,m,s)$ is minimised at $m = (n-s)/2$, when $T(n,s) < \beta$, and is minimised at $m = 0$, when $T(n,s) > \beta$.  }

We now turn to the construction of networks that are optimal for the hider.  Fix the number of singleton nodes, $s\leq n-4$.
Define a new function $\bar{Q}(n,s)$ as follows 
\begin{equation*}
\bar{Q}(n,s) = \begin{cases}
Q(n,0,s), & \textrm{if $0\leq s \leq n-4$ and $T(n,s) \geq \beta$,} \\
Q(n,(n-s)/2,s), & \textrm{if $0\leq s \leq n-4$, $T(n,s) < \beta$ and $n-s$ is even,} \\
Q(n,(n-s-3)/2,s), & \textrm{if $0\leq s \leq n-4$, $T(n,s) < \beta$ and $n-s$ is odd.}
\end{cases}
\end{equation*}

Consider first the case where $n-s$ is even. 
\begin{lemma} 
\label{even}
Suppose $H$ builds a network with $s$ singleton nodes such that $n-s$ is even. Then, an optimal strategy for $H$ provides $H$ with payoff $-\bar{Q}(n,s)$.
If $T(n,s) < \beta$, $G$ is optimal if  the subnetwork over $n-s$ nodes is a maximal core-periphery network. If $T(n,s) > \beta$,  $G$ is optimal if  the subnetwork over $n-s$ nodes is a cycle.
\end{lemma}

\begin{proof} Fix $s$ such that $n-s$ is even. Let

\begin{equation*}\bar{A}(n,s) = \begin{cases}
               A(n,(n-s)/2,s), & \textrm{if $T(n,s) < \beta$,} \\
               A(n,0,s), & \textrm{if $T(n,s) \geq \beta$.}
               \end{cases}
\end{equation*}
and let
\begin{equation}
\label{eq:kappa}
\kappa= \begin{cases}
        \frac{B(s)+f(1)}{\bar{A}(n,s)+B(s)+f(n-s)+f(1)} & \textrm{if $\bar{A}(n,s)  > -f(1)$,}\\
        1, & \textrm{otherwise.}
        \end{cases}
\end{equation}
Let $H$ choose a network $G$ such that :
\begin{itemize}
\item[(i)] $G$ has exactly $s$ singletons.
\item[(ii)] $G$ is a maximal core periphery on $n-s$ nodes if  $T(n,s) < \beta$.
\item[(iii)] $G$ is a cycle on $n-s$ nodes if  $T(n,s)\geq \beta$.
\end{itemize}

Moreover,  suppose that  the hider hides in the component of size $n-s$ with probability $\kappa$, mixing uniformly on the periphery nodes in the case of the component being a core-periphery network,  and mixing uniformly over all its nodes in the case of the component being a cycle.  Also,  she hides in the singleton nodes with probability $1-\kappa$,  mixing uniformly on them. By similar arguments to those used for $\lambda_{\mathrm{S}}$ above, $\kappa \in [0,1]$ and so the strategy is valid.
If the seeker seeks in the singleton nodes, this yields payoff of at least $\kappa f(n-s) - (1-\kappa)B(s)$ to the hider. Similarly, if the seeker seeks in the core-periphery component, this yields payoff of at least $-\kappa  \bar{A}(n,s) + (1-\kappa) f(1)$ to the hider. With the value of $\kappa$, above, both values are equal in the case of $\bar{A}(n,s) > -f(1)$, and the latter is greater, otherwise. 

Hence, the strategy guarantees a payoff $-\kappa \bar{A}(n,s) + (1-\kappa) f(1)$ to the hider. Note that
\[
 -\kappa \bar{A}(n,s) + (1-\kappa) f(1) = - \bar{Q}(n,s)\]
 
 Recall that we have shown that $\bar{Q}(n,s)$ is the minimal payoff the seeker can get on any network with exactly $s$ singleton nodes. Since the game is zero-sum, $-\bar{Q}(n,s)$ is the maximal payoff the hider can get on any network with exactly $s$ singleton nodes and hence the network constructed above as well as the hiding strategy must be optimal for the hider.  
\end{proof}

Next, consider the case of $n-s$ being odd. 

\begin{lemma}
\label{odd}  
Suppose that $n-s$ is odd.  Then, an optimal strategy for $H$ gives him a payoff of $-Q(n,(n-s-3)/2,s)$.
If $T(n,s) < \beta$, $G$ is optimal if  the subnetwork over $n-s$ nodes is a maximal core-periphery network. If $T(n,s) > \beta$,  $G$ is optimal if it  the subnetwork over $n-s$ nodes is a cycle..
\end{lemma}

\begin{proof}
Let 
\begin{equation*}
\bar{A}(n,s) = \begin{cases}
               A(n,(n-s-3)/2,s), & \textrm{if $T(n,s) < \beta$,} \\
               A(n,0,s), & \textrm{if $T(n,s) \geq \beta$.}
               \end{cases}
\end{equation*}
and let $\kappa$ be defined as in~\eqref{eq:kappa}.
If $T(n,s) \geq \beta$ than choosing a cycle over $n-s$ nodes and using the same hiding strategy as in the case of $n-s$ being even, the hider secures the highest possible payoff on a network with exactly $s$ singleton nodes. 
Suppose that $T(n,s) < \beta$. Since  $(n-s)/2$ is not an integer,  the hider cannot attain the upper bound on his payoff determined by the lower bound on the payoff to the seeker, $\bar{Q}(n,s)$. Recall that if $T(n,s) < \beta$ then for any $0 \leq s \leq n-4$, $Q(n,m,s)$ is decreasing in $m$. We show  below for any $0\leq s \leq n-4$,  the hider can attain payoff $-Q(n,(n-s-3)/2,s)$,  and that this is the maximal payoff he can secure when $n-s$ is odd.

Suppose that the hider chooses a maximal core-periphery network (with three orphaned nodes) over $n-s$ nodes (c.f. Figure~\ref{fig:cp3orph}).

\begin{center}
\begin{figure}[!htb]
\centering
\includegraphics[scale=.75]{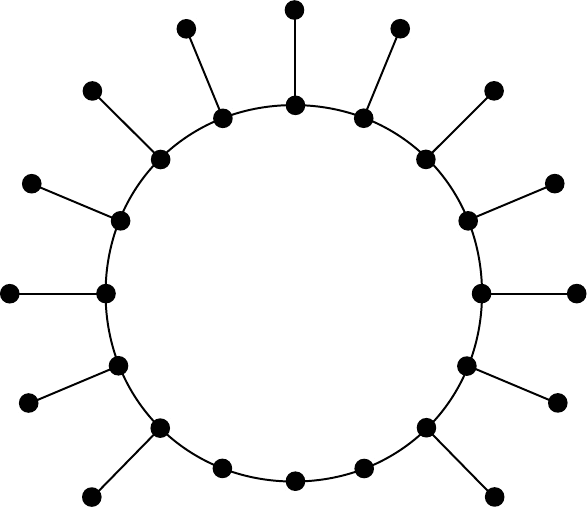}
\caption{A core-periphery network over $23$ nodes with $3$ orphaned nodes.}
\label{fig:cp3orph}
\end{figure}
\end{center}

Consider a strategy of the hider
\begin{equation*}
\bm{\eta} = \kappa (\mu \bm{\eta}^{\mathrm{M}} + (1-\mu) \bm{\eta}^{\mathrm{R}}) + (1 - \kappa)\bm{\eta}^{\mathrm{S}},
\end{equation*}
where
\begin{equation*}
\eta_i^{\mathrm{M}} = \begin{cases}
                           \frac{1}{m}, & \textrm{if $i \in S\!L(G)$,} \\
                           0, & \textrm{otherwise,}
                           \end{cases}
\end{equation*}
(i.e. $\bm{\eta}^{\mathrm{M}}$ mixes uniformly on the periphery nodes of $G$),
\begin{equation*}
\eta_i^{\mathrm{R}} = \begin{cases}
                           1, & \textrm{if $i$ is the middle orphaned node in $G$,} \\
                           0, & \textrm{otherwise,}
                           \end{cases}
\end{equation*}
\begin{equation*}
\eta_i^{\mathrm{S}} = \begin{cases}
                           \frac{1}{s}, & \textrm{if $i \in S(G)$,} \\
                           0, & \textrm{otherwise.}
                           \end{cases}
\end{equation*}
(i.e. $\bm{\eta}^{\mathrm{S}}$ mixes uniformly on the singleton nodes of $G$), and
\begin{equation*}
\mu = \frac{(n - s - 3)f(n-s-2) + (n - s - 3)\beta}{(n - s - 3)f(n-s-1) + 2f(n-s-2) + (n - s - 1)\beta}.
\end{equation*}
It is immediate to see that $\mu \in [0,1]$ and so the hiding strategy is valid. 
If the seeker seeks in the orphaned nodes of the core-periphery component, this yields payoff of at least $\kappa (\mu f(n-s-1) - (1-\mu)\beta) + (1-\kappa)f(1)$ to the hider and, since the game is zero-sum, of at most minus this value to the seeker. Similarly, if the seeker seeks in periphery nodes or their neighbours in the core-periphery component, this yields payoff of at least $\kappa (\mu (-2\beta/(n-s-3) + (1-2/(n-s-3))f(n-s-2)) + (1-\mu)f(n-s-2)) + (1-\kappa)f(1)$  to the hider and of at most minus this value to the seeker. With the value of $\mu$, above, both these guarantees are equal. 

It is  straightforward  to verify that
\begin{eqnarray*}
 \kappa (\mu f(n-s-1) - (1-\mu)\beta) + (1-\kappa)f(1)&= &-\kappa A(n,(n-s-3)/2,s) + (1-\kappa) f(1) \\&= &-Q(n,(n-s-3)/2,s).\end{eqnarray*}

Since $Q(n,(n-s-3)/2,s)$ is a lower bound on the payoff that the seeker can secure in a network with exactly $s$ singleton nodes and at most $(n-s-3)/2$ singleton leaves, the negative of  this value is the highest payoff that the hider can secure in a network with exactly $s$ singleton nodes and at most $(n-s-3)/2$ singleton leaves. The only networks that could yield a higher payoff to the seeker are networks with exactly $s$ singleton nodes and $(n-s-1)/2$ singleton leaves. But we show in Lemma \ref{lemma:orphan} in the appendix that these networks have a lower value for the hider.
\end{proof}

Since the game is zero-sum, the hider maximises his payoff when the seeker's payoff is minimised. Therefore,  an optimal network has $s \in S^{*}(n)$ singleton nodes, where
\begin{equation*}
S^{*}(n) = \arg\min_{s \in \{0,\ldots,n\}} \bar{Q}(n,s).
\end{equation*}

Lemmas~\ref{even} and~\ref{odd} have therefore proved the characterization result that we summarize in the following Theorem.

\begin{theorem}
\label{th:singlemain}
For any number of nodes, $n \geq 1$, and any $\beta \geq 0$ there exists an equilibrium of the game, $((G,h), s)$ such that
\begin{itemize}
\item $G$ has exactly $s\in S^{*}(n)$ singleton nodes and either $s \leq n-4$ or $s = n$.
\item If $T(n,s) \geq \beta$ and $n - s \geq 4$ then $G$ has a cycle component over the remaining $n-s$ nodes.
\item If $T(n,s) < \beta$, $n - s \geq 4$ then $G$ has a maximal core-periphery component over $n-s$ nodes.
\item The hider mixes between hiding in the singleton nodes and in the connected component. When hiding in the singleton nodes, he mixes uniformly across all these nodes. When hiding in the connected component, he mixes uniformly across all the nodes (when it is a cycle), mixes uniformly across the periphery nodes (when it is a maximal core-periphery network and $n-s$ is even), and mixes between hiding in periphery nodes, mixing uniformly across them, and the middle orphaned node (otherwise).
\item The seeker mixes between seeking in the singleton nodes and in the connected component. When seeking in the singleton nodes, he mixes uniformly across all these nodes. When seeking in the connected component, he mixes uniformly across all the nodes (when it is a cycle), mixes uniformly across the core nodes (when it is a maximal core-periphery network and $n-s$ is even), and mixes between seeking in the neighbours of periphery nodes, mixing uniformly across them, and the middle orphaned node (otherwise).
\end{itemize}
Equilibrium payoff to the hider is $-\bar{Q}(n,s)$.
\end{theorem}

We have shown  in the proof of Theorem~\ref{th:singlemain}, that the equilibrium payoff to the seeker in an optimal network with at least one singleton node is a convex combination of $B(s)$ which is greater than $-f(1)$) and $-f(1)$ and so it is at least $-f(1)$. Hence the payoff that the hider can secure in such a network is at most $f(1)$. Thus if the payoff the seeker can secure in a connected component of size $n$, $\bar{A}(n,0)$ is smaller than $ -f(1)$, then the payoff the hider can secure in such a component is $-\bar{A}(n,0) > f(1)$. If that inequality holds,  it is optimal for the hider to choose a connected network without singleton nodes. 

If, on the other hand, the cost of being caught, $\beta$, is sufficiently high then $\bar{A}(n,0) > -f(1)$ and the payoff the hider can secure in a connected network, $-\bar{A}(n,0)$, is less than the payoff he gets if he is not caught in a singleton node. This leads the hider to construct a network with a smaller component and $s\geq 1$ singleton nodes. If the cost of being caught is sufficiently high, it is optimal for the hider to choose a disconnected network with $s = n$ singleton nodes.

The characterization of equilibrium networks provided in Theorem~\ref{th:singlemain} is not complete. This Theorem displays network architectures which achieve the highest possible payoff for the hider, but does not show that these network topologies are unique. As  we prove below, if $T(n,s) < \beta$ the connected component must be a maximal core-periphery network. So in this case we obtain complete characterization of equilibrium networks.If $T(n,s) > \beta$ there exist network topologies other than the cycle which are optimal. We establish necessary properties that the optimal network topologies must possess.

\begin{theorem}
\label{th:singlecomplete}
For any number of nodes, $n \geq 1$, and any $\beta \geq 0$, if $((G,h), s)$ is an equilibrium of the game then
\begin{itemize}
\item $G$ has exactly $s\in S^{*}(n)$ singleton nodes.
\item If $T(n,s) < \beta$, $n - s \geq 4$, then $G$ has a maximal core-periphery component over $n-s$ nodes.
\item If $T(n,s) > \beta$, then $G$ has a $2$-connected component over $n-s$ non-singleton nodes with at least $\lceil (n-s)/3 \rceil$ nodes of degree $2$ and the hider never hides in nodes of degree greater than $2$ in equilibrium.
\end{itemize}
\end{theorem}

\begin{proof}
The fact that $G$ must have exactly $s \in S^{*}(n)$ singleton leaves is already established in proof of Theorem~\ref{th:singlemain}.
For the properties of the remaining part of equilibrium network, we consider the cases of $T(n,s) < \beta$ and $T(n,s) > \beta$ separately.

Suppose that $T(n,s) < \beta$. Suppose first that $n - s$ is even. 
Since $Q(n,m,s)$ is decreasing in $m$ and the maximum feasible value for $m$, when $n-s$ is even, is $(n-s)/2$ so the subnetwork over $n-s$ non-singleton nodes in any optimal network must have $(n-s)/2$ singleton leaves. If the network is optimal, the neighbours of the singleton leaves must form a $2$-connected network. Otherwise, the seeker would obtain a payoff that is strictly higher than $\bar{Q}(n,s)$ by mixing uniformly on the neighbours of non-singleton leaves  when seeking outside singleton nodes. This is because in the case of not capturing the hider, he will leave the subnetwork over $n-s$ nodes disconnected with probability greater than $0$. Hence the optimal subnetwork over $n-s$ non-singleton nodes  must be a maximal core-periphery network.
Second, suppose that $n - s$ is odd. As we showed in proof of Lemma~\ref{odd}, the optimal number of singleton leaves in the subnetwork over $n-s$ non-singleton nodes is $(n-s-3)/2$. Moreover, as we argued above, nodes which are not singleton leaves must form a $2$-connected network. Thus this subnetwork must be a core-periphery network with $2$-connected core and three orphaned nodes. What remains to be shown is that one of the orphaned nodes must have exactly the other two orphaned nodes as its neighbours in this subnetwork. Since the subnetwork formed by the nodes of the core must be $2$-connected, any node of the core must have at least two neighbours. Suppose, to the contrary, that each of the orphaned nodes has at least one neighbour that is not an orphaned node. Then, mixing uniformly on non-orphaned core nodes, the seeker captures the hider with higher probability than in a maximal core-periphery network (regardless of the strategy of the hider) and causes the same damage in the case of not capturing the hider. This results in strictly lower payoff to the hider than $-\bar{Q}(n,s)$ and so the network is not optimal. Therefore the neighbours of one of the orphaned nodes must be exactly the two other orphaned nodes.

Suppose next that $T(n,s) > \beta$. In this case, $Q(n,m,s)$ is increasing in $m$ and so the  optimal network has no singleton leaves in the subnetwork over the $n-s$ non-singleton nodes.
Let $U$ be the set of $n-s$ non-singleton nodes in the network and let $F$ be the subnetwork over this set of nodes.
As  argued above, the seeker has a seeking strategy that guarantees him a probability of capture at least $3/(n-s)$ in $F$. If $F$ is not $2$-connected,   the seeker will leave the subnetwork disconnected in the event of not capturing $H$.  This  gives strictly lower payoff to $H$ than in the cycle.  Hence $F$ must be $2$-connected. Hence  all the nodes in $F$ have degree at least $2$. 
Suppose that $F$ has $t < \lceil (n-s)/3 \rceil$ nodes of degree 2.  Note that since $F$ is 2-connected, only one node is removed if $H$ is not captured. So, the expected payoff of $H$ (and hence $S$) only depends on the probability of capture. 
Consider any strategy $\bm{\eta}$ of $H$ and let $T$ be its support on $U$.  Let $\bm{\sigma}'_T$ be a mixed strategy of the seeker that mixes uniformly on $N_F[T]$.\footnote{
Given graph $G = \langle V, E \rangle$ and a set of nodes $U\subseteq V$, set $N_G[U] = U \cup \{v \in V : uv \in E \textrm{ for some } u \in U\}$ is the \emph{closed neighbourhood} of $U$ in $G$.
}   Let $\bm{\sigma}_T = \lambda \bm{\sigma}'_T + (1-\lambda) \bm{\sigma}^{\mathrm{S}}$ be a strategy of the seeker that mixes uniformly on the singleton nodes with probability $1-\lambda$ and uses $\bm{\sigma}'_T$ with probability $\lambda$, where $\lambda$ is such that the lower bound on the expected payoff to the seeker when the hider hides in $T$ is equal to the lower bound on the expected payoff to the seeker when the hider hides in singleton nodes.
Notice that the lower bound on the expected payoff to the seeker from using $\bm{\sigma}'_T$ when the hider hides in $T$ is strictly higher than $3/(n-s)$. For if $T$ contains a node of degree at least $3$ then the seeker captures the hider with probability strictly greater than $3/|N_F[T]| \geq 3/(n-s)$, and if $T$ does not contain a node of degree $3$ then $|N_F[T]| \leq 3|T| < n-s$ and the seeker captures the hider with probability $3/|N_F[T]| > 3/(n-s)$. Hence there exists $p_T > \bar{Q}(n,s)$ such that the expected payoff to the seeker from using $\bm{\sigma}_T$ against any strategy of the hider, $\bm{\eta}$, with support $T$ on $U$ is at least $p_T$. Taking $\varepsilon = \min_{T\subseteq U} (\bar{Q}(n,s) - p_T)$ shows that $F$ cannot be optimal forv $H$. 
Notice also that if the support of $H$'s strategy in a network with $2$-connected component $F$ contains nodes of degree greater then $3$ then strategy $\bm{\sigma}$ guarantees the seeker payoff strictly greater than $\bar{Q}(n,s)$. Therefore, in equilibrium, the hider never hides in nodes of degree greater than $2$ in the $2$-connected component of an optimal network.
\end{proof}

We next provide examples of topologies of the connected component other than the cycle in equilibrium networks for the case of $T(n,s) > \beta$. Suppose that $n-s = 3t$ where $t \geq 2$ is an integer. Let $U$ be the set of nodes of the component. Suppose that the nodes in $U$ are connected, forming a cycle, and let $T\subseteq U$, $|T| = t$,  be a subset of the nodes such that any two nodes in $T$ are separated by two nodes from $U\setminus T$. Any network obtained from from the cycle by adding links between the nodes in $U\setminus T$ is optimal (an example is presented in Figure~\ref{fig:noncycle}). Both players mixing uniformly on $U$ is an equilibrium on any such network.

\begin{center}
\begin{figure}[!htb]
\centering
\includegraphics[scale=.75]{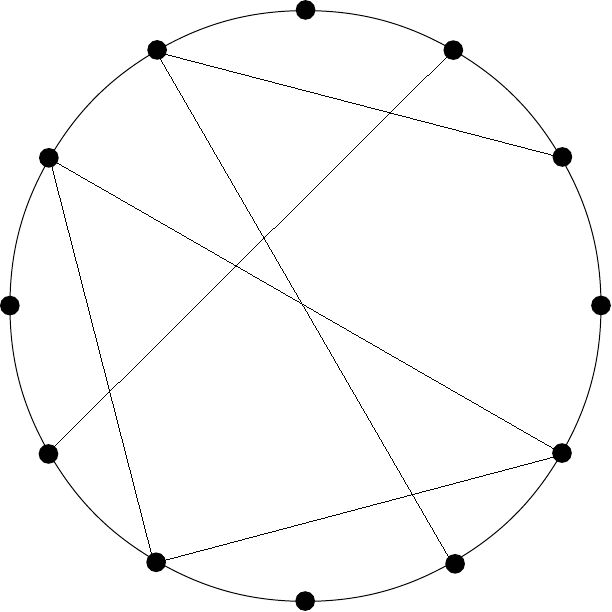}
\caption{An optimal component for $n-s = 12$.}
\label{fig:noncycle}
\end{figure}
\end{center}

Theorems~\ref{th:singlemain} and~\ref{th:singlecomplete} provide  a characterization of optimal networks for the hider in terms of the quantity $T(n,s)$. As this expression is not transparent, we provide sufficient conditions on the utility function $f(\cdot)$ which guarantee that the connected component of an optimal network is a maximal core-periphery network.

 \begin{theorem}
 \label{functf}
 Suppose that either 
 \begin{itemize}
 \item[(i)] $f$ is concave, or
 \item[(ii)] $f$ is convex and for all $x \ge 2$
 \[f(x+1) < \frac{x}{x-1} f(x)\]
 \end{itemize}
  Then, for all $n \geq 1$, and any $\beta \geq 0$, $G$ is an equilibrium network if and only if $G$ has $s \in S^{*}(n)$  singleton nodes and a maximal core-periphery component over $n-s$ nodes.
In addition, if $f$ is linear then $S^*(n) = \{0,1,n\}$.
\end{theorem}

\begin{proof}
Notice that 
\begin{align*}
T(n,s) = (n-s-3)\Delta f(n-s-2) - f(n-s-2)
\end{align*}
and 
\begin{align*}
T(n,s+1) = (n-s-3)\Delta f(n-s-3) - f(n-s-2)
\end{align*}
Hence, 
\begin{align*}
T(n,s+1) - T(n,s) & = -(n-s-3)(\Delta f(n-s-2) - \Delta f(n-s-3)) \\
                  & = -(n-s-3)\Delta^2 f(n-s-3). 
\end{align*}
where $\Delta f(x) = f(x+1) - f(x)$ is the first-order (forward) difference of $f$ at $x$ and $\Delta^2 f(x) = \Delta f(x+1) - \Delta f(x)$ is the  second-order (forward) difference of $f$ at $x$.
Hence,  if $f$ is concave, then $\Delta^2 f(n-s-3) \leq 0$, and  so
\[ T(n,s+1) - T(n,s) \geq 0 \mbox{ for all  }s \leq n-4\]
 In addition $T(n,n-4) = f(3) - 2f(2)$ which is negative if $f$ is concave and strictly increasing. 
 Thus for all $n \geq 4$ and $s \leq n-4$, $T(n,s) < 0 \leq \beta$.
 
From Theorems~\ref{th:singlemain} and~\ref{th:singlecomplete}, $G$ is an equilibrium network if and only if its connected component is a maximal core-periphery network over $n-s$ nodes.
  
If $f$ is convex then $\Delta^2 f(n-s-3) \geq 0$ and $T(n,s+1) - T(n,s) \leq 0$, for all $s \leq n-4$. Thus $T(n,s)$ is decreasing in $s$ on $[0,n-4]$, for all $n\geq 4$.

 Suppose that $f(x+1) < x/(x-1) f(x)$ for all $x \geq 2$.\footnote{
An example of a family of strictly increasing convex functions that satisfy this property are the functions $f(x) = x^\gamma/(x+1)^{\gamma - 1}$ with $\gamma > 1$.
}
 Then $T(n,0) = (n-3)f(n-1) - (n-2)f(n-2) < 0$ and so
  \[T(n,s) \leq T(n,0) < \beta, \mbox{ for all } s \in [0,n-4]. \]
  
Again, by Theorems~\ref{th:singlemain} and~\ref{th:singlecomplete}, $G$ is an equilibrium network if and only if its connected component is a maximal core-periphery network over $n-s$ nodes.

Next, note that if $n \le 5$, then Lemma~\ref{sing} shows that $ s^*\le1 $. 
Suppose that $f$ is linear and that $n \geq 6$. We show in the Appendix (Lemma~\ref{lemma:minpayoffs}) that if $n\geq 6$, then $Q(n,(n-s)/2,s)$ is minimised either at $s = 0$ or at $s = 1$ or at $s = n$. 
This shows that $s^* \in \{0, 1, n\}$ and completes the proof of the theorem.
\end{proof}

\begin{remark} 
The theorem establishes a full characterization of equilibrium networks when $f$ is concave or convex but growing slowly.
\end{remark}

\section{Conclusions}
We proposed and studied a strategic model network design and hiding in the network facing a hostile authority that attempts to disrupt the network and capture the hider. We characterized optimal networks for the hider as well as optimal hiding and seeking strategies in these networks. Our results suggests that the hider chooses networks that allow him to be anonymous and peripheral in the network. We also developed a technique for solving such models in the setup of zero-sum games.

There are at least two avenues for future research. Firstly, different forms of benefits from the network could be considered. For example,  the utility of the hider could dependent not only on the size of his component but also on his distance to the nodes in the component. Given our results, we conjecture that this would make the core periphery components with better connected core more attractive. But answering this problem precisely requires formal analysis.
Secondly, the seeker could be endowed with more than one seeking unit and the units could be used either simultaneously or sequentially. Our initial investigation suggests that solving such an extension might be an ambitious task. 

\newpage
\bibliographystyle{plainnat}
\bibliography{biblio}
\newpage
\appendix
\section*{Appendix}

\section{Proofs}

\begin{lemma}
\label{lemma:orphan}
If $n-s$ is odd and $T(n,s) < \beta$, the hider obtains a higher expected payoff in a core-periphery network with $(n-s-3)/2$ singleton leaves than in a core-periphery network with $(n-s-1)/2$ singleton leaves.
\end{lemma}

\begin{proof}
In a core-periphery network with $(n-s-1)/2$ singleton leaves,  the set $R(G)$ consist of exactly one node and this node is connected to at least two nodes in $M(G)$. It cannot be connected to one node in $M(G)$, because in this case its neighbour would have two leaf-neighbours and could not be a member of $M(G)$. 

Let $\bm{\tilde{\sigma}} = \lambda \bm{\sigma}^{\mathrm{S}} + (1-\lambda) \bm{\sigma}^{\mathrm{M}}$, where $\bm{\sigma}^{\mathrm{M}}$ and $\bm{\sigma}^{\mathrm{S}}$ are the mixed strategies of the seeker, defined earlier in the proof,
\begin{equation*}
\lambda = \begin{cases}
          \frac{X(n,s) + f(1)}{B(s) + X(n,s) + f(1) + f(n-s)}, & \textrm {if $X(n ,s) > -f(1)$,} \\
          0, & \textrm{otherwise,}
          \end{cases}
\end{equation*}
and
\begin{equation*}
X(n,s) = \frac{2\beta}{n-s-1} - \left(1 - \frac{2}{n-s-1}\right)f(n-s-2).
\end{equation*}
Using this strategy, with probability $\lambda$, $S$ mixes uniformly on the nodes in $M(G)$ and with probability $(1-\lambda)$, $S$ mixes uniformly on the singleton nodes of $G$. The payoff to $S$ conditional on $H$ hiding in a singleton node is at least $\lambda B(s) - (1-\lambda)f(1)$ and the payoff to $S$ conditional on $H$ hiding outside singleton nodes is at least $(1-\lambda)X(n,s) - \lambda f(n-s)$.
It is easy to verify that the value of $\lambda$ is such that both these payoffs are equal (in the case of $X(n,s) > -f(1)$) or the latter is higher, for any value of $\lambda$. Therefore the payoff to $S$ from using $\bm{\tilde{\sigma}}$ against any strategy of $H$ is at least
\begin{equation*}
Y(n,s) = \begin{cases}
\frac{B(s) X(n,s) - f(1) f(n-s)}{B(s) + X(n,s) + f(1) + f(n-s)}, & \textrm{if $X(n,s) > -f(1)$,} \\
X(n,s), & \textrm{otherwise,}
\end{cases}
\end{equation*}
and so the upper bound on the payoff to the hider on any network with $s$ singleton nodes and $(n-s-1)/2$ singleton leaves is at most $-Y(n,s)$. To see that $-Q(n,(n-s-3)/s,s) > -Y(n,s)$ notice that 
\begin{multline*}
X(n,s) - A(n,(n-s-3)/2,s) = \\
 \frac{2(f(n-s-1) - f(n-s-2))(f(n-s-2) + \beta)(n - s - 3)}{(n - s - 1)(f(n-s-1)(n - s - 3) + 2f(n-s-2) + \beta(n - s - 1))} > 0
\end{multline*}
and so $X(n,s) > A(n,(n-s-3)/2,s)$. 

Next consider the following Claim:

\begin{claim}
\label{claim:decr}
The function
\begin{equation*}
\varphi(Z) = \begin{cases}
        \frac{B(s)Z - f(1)f(n-s)}{Z + B(s) + f(n-s) + f(1)}, & \textrm{if $Z > -f(1)$,} \\
        Z, & \textrm{otherwise,}
        \end{cases}
\end{equation*}
is strictly increasing in $Z$.
\end{claim}

\begin{proof}
Notice that $\varphi(-f(1)) = -f(1)$ when $Z = -f(1)$. Moreover, $\varphi$ is increasing in $Z$ if $Z < -f(1)$.
Let $Z > -f(1)$. Taking the derivative of $\varphi$ with respect to $Z$ we get
\begin{equation*}
\varphi'(Z) = \frac{(B(s)+f(1))(B(s)  + f(n-s))}{(Z + B(s) + f(n-s) + f(1))^2}
\end{equation*}
and it is immediate to see that $\varphi'(Z) > 0$ and $\varphi$ increases in $Z$ when $B(s) > -f(1)$ and $B(s) \geq -f(n-s)$.
Notice that $B(s) = (\beta+f(1))/s -f(1) > -f(1)$ for any $\beta \geq 0$ and $s > 0$. Also $f(n-s) \geq f(1)$ for all $s \in [0,n-1]$.
Thus, by the observation on function $\varphi$, above, $\varphi(Z)$ increases when $Z$ increases. 
\end{proof}

Claim \ref{claim:decr}, together with $X(n,s) > A(n,(n-s-3)/2,s)$, implies that $Y(n,s) > Q(n,(n-s-3)/2,s)$, completing the proof of the Lemma.  
\end{proof}

\begin{lemma}
\label{lemma:minpayoffs}
Let $\lambda > 0$ and let $f(x) = \lambda x$, for all $x\in \mathbb{R}_{\geq 0}$.
For any natural $n \geq 6$, $t \in \{0,1\}$ and any $s \in \{t+1,\ldots,n\}$, $Q(n,(n-s)/2,s) > \min(Q(n,0,n),Q(n,(n-t)/2,t))$
\end{lemma}

\begin{proof}
Let $f(x) = \lambda x$, with $\lambda > 0$, and let $\widetilde{\beta} = \beta / \lambda$. Let
\begin{align*}
\widetilde{A}(n,s) & = A(n,(n-s)/2,s) = \lambda\left(2\left(\frac{\widetilde{\beta}-2}{n-s}\right) + 4 - (n-s)\right), \textrm{ for $0 \leq s \leq n-2$,} \\
B(s) &= \lambda\left(\frac{\widetilde{\beta} + 1}{s} - 1\right),
\end{align*}
and
\begin{equation*}
\widetilde{Q}(n,s) = Q(n,(n-s)/2,s) = \begin{cases}
               \widetilde{A}(n,s), & \textrm{if $\widetilde{A}(n,s) \leq -\lambda$ or $s = 0$,} \\
               A\!B(n,s), & \textrm{if $1 \leq s \leq n-2$ and $\widetilde{A}(n,s) > -\lambda$}\\
               B(n), & \textrm{otherwise,}
               \end{cases}
\end{equation*}
with
\begin{equation}
\label{eq:ab}
A\!B(n,s) = (1-\rho)\widetilde{A}(n,s) - \rho \lambda (n-s)
\end{equation}
where $\rho$ solves
\begin{equation}
\label{eq:varrho}
(1-\rho)\widetilde{A}(n,s) + \rho\lambda (s-n) = \rho B(s) - (1 - \rho)\lambda.
\end{equation}

Solving~\eqref{eq:varrho} we get
\begin{equation*}
\rho = \frac{s(2(\widetilde{\beta} - 2) - (n-s)(n-s-5))}{s(2(\widetilde{\beta} - 2) - (n - s)(n - s - 5)) + (n - s)(s(n - s - 1) + \widetilde{\beta} + 1)}.
\end{equation*}
Notice that $2(\widetilde{\beta} - 2) - (n-s)(n-s-5) > 0$ if and only if $\widetilde{A}(n,s) > -\lambda$, and $(n - s)(s(n - s - 1) + \widetilde{\beta} + 1) > 0$ for $s \leq n-1$.
Thus if $\widetilde{A}(n,s) > -\lambda$ then $\rho \in (0,1)$. In addition $B(s) > -\lambda$, for all $s > 0$, so if $\rho \in (0,1)$ then $A\!B(n,s) > -\lambda$. Moreover, $\widetilde{A}(n,s)$ is increasing in $s$ on $[0,n-2]$ and it is equal to $\beta$ at $s = n-2$. 
By the observations above, if $\widetilde{A}(n,1) \leq -\lambda$ then $\bar{Q}(n,0) = \widetilde{A}(n,0) < \widetilde{A}(n,1) = \widetilde{Q}(n,1) \leq -\lambda < \bar{Q}(n,s)$, for all $s \in \{2,\ldots,n\}$, and the claim of the lemma holds.

For the remaining part of the proof suppose that $\widetilde{A}(n,1) > -\lambda$. This implies $2(\widetilde{\beta}-2) > (n-1)(n-6)$ and, consequently, $\widetilde{\beta} > 2$ if $n \geq 6$. 
We will show that $\widetilde{Q}(n,s)$ is either decreasing or first increasing and then decreasing on $[0,n-1]$. Let $\tilde{s} = \inf \{s \in [0,n-2) : \widetilde{A}(n,s) \geq -\lambda\}$. Since $\widetilde{A}(n,s)$ is increasing in $s$ and equal to $\beta \geq 0$ at $s = n-2$ so the infimum exists and $\tilde{s}$ is well defined. On $[0,\tilde{s}]$, $\widetilde{Q}(n,s) = \widetilde{A}(n,s)$ and, as we argued above, $\widetilde{Q}(n,s)$ is increasing. Consider the interval $[\tilde{s},n-1]$. Notice that since $B(s) > -\lambda \geq -\lambda(n-s)$, for all $0 < s \leq n-1$, and $\widetilde{A}(n,\tilde{s}) = -\lambda$ so $A\!B(n,\tilde{s}) = -\lambda$. In addition, $A\!B(n,n) = B(n)$.
We will show that $A\!B(n,s)$ is either decreasing or first increasing and then decreasing on $[0,n]$. Inserting $\rho$ into~\eqref{eq:ab} we get
\begin{equation*}
A\!B(n,s) = \frac{(n^2(\widetilde{\beta} + 1) - 2n(s(\beta/\lambda - 1) + 2(\widetilde{\beta} + 1)) + s^2(\widetilde{\beta} - 3) + 6s\widetilde{\beta} - 2(\widetilde{\beta} + 1)(\widetilde{\beta} - 2))}{s(4s - \widetilde{\beta} + 5) - n(4s + \widetilde{\beta} + 1)}.
\end{equation*}
Taking the derivative of $A\!B(n,s)$ with respect to $s$ we get
\begin{equation*}
\frac{\partial A\!B(n,s)}{\partial s} = \frac{(\widetilde{\beta} + 1) W(s)}{(s(4s - \widetilde{\beta} + 5) - n(4s + \widetilde{\beta} + 1))^2},
\end{equation*}
where
\begin{equation*}
W(s) = X s^2 - 2Ys + \left(n+\frac{\widetilde{\beta}-2}{2}\right) Y - \left(\frac{\widetilde{\beta}-2}{2}\right)(n-4)(\widetilde{\beta}+1),
\end{equation*}
with $X = 4n - \widetilde{\beta} - 15$ and $Y = 4n^2 + n(\widetilde{\beta} - 19) - 8(\widetilde{\beta} - 2)$.

The sign of $\partial A\!B/\partial s$ is the same as the sign of $W(s)$. Notice that $W(n) = -2(\widetilde{\beta} - 2)(n+\widetilde{\beta}-5)< 0$, as $n\geq 6$ and $\widetilde{\beta} > 2$. When $X > 0$, then $W(s)$ is an $\bigcup$-shaped parabola and, since $W(n) \leq 0$, either $W$ is negative or $W$ is first positive and the negative on $[0,n]$. Thus in this case $A\!B$ is either increasing or first increasing and then decreasing on $[0,n]$.
Similar observation holds when $X = 0$. Suppose that $X < 0$. In this case  $W(s)$ is an $\bigcap$-shaped parabola and it has a maximum at $s^{*} = Y/X$. Suppose that $s^{*} \in (0,n-2)$. Since $X < 0$ so $Y < 0$. Moreover, for $n \geq 6$, $X < 0$ implies $\beta > 5$ and, consequently, 
\begin{align*}
W(s^{*}) & =  - Ys^{*} +  \left(n+\frac{\widetilde{\beta}-2}{2}\right) Y - \left(\frac{\widetilde{\beta}-2}{2}\right)(n-4)(\widetilde{\beta}+1) \\
         & = \left(n - s^{*} +\frac{\widetilde{\beta}-2}{2}\right) Y - \left(\frac{\widetilde{\beta}-2}{2}\right)(n-4)(\widetilde{\beta}+1) < 0.
\end{align*}
Thus $W$ is either negative or first positive then negative on $[0,n]$, for any natural $n\geq 5$. Hence $A\!B{Q}$ is either decreasing
or first increasing and then decreasing on $[0,n]$, for any natural $n\geq 6$.

By the analysis above, when $\widetilde{A}(n,1) > -\lambda$ then $A\!B(n,s)$ is either decreasing or first increasing and then decreasing in $s$ on $[0,n]$ and $A\!B(n,n) = B(n)$. Hence, by the definition of $\widetilde{Q}(n,s)$, the claim of the lemma follows immediately.
\end{proof}

\newpage
















\end{document}